\documentclass[a4paper,twocolumn,10pt,accepted=2023-02-01]{quantumarticle}
\pdfoutput=1
\usepackage{exscale}
\usepackage{bbm}
\usepackage{graphicx}
\usepackage{latexsym}
\usepackage[T1]{fontenc}
\usepackage{enumerate}
\usepackage{bbold}
\usepackage{color}
\usepackage{mathtools}
\usepackage{changes}
\usepackage[normalem]{ulem}
\usepackage{amsfonts,amsmath,amssymb,amsthm}
\usepackage[colorlinks=true,citecolor=blue,urlcolor=blue]{hyperref}
\usepackage[cm]{fullpage}

\usepackage{cancel}

\newcommand{\ket}[1]{|#1\rangle}
\newcommand{\bra}[1]{\langle#1|}

\theoremstyle{plain}
\newtheorem{thm}{Theorem}
\newtheorem{lem}{Lemma}

\newtheorem{con}{Conjecture}
\newtheorem{cor}{Corollary}

\begin{document}

\title{Fully non-positive-partial-transpose genuinely entangled subspaces}

\author{Owidiusz Makuta}
\affiliation{Center for Theoretical Physics, Polish Academy of Sciences, Aleja Lotnik\'{o}w 32/46, 02-668 Warsaw, Poland}

\author{Błażej Kuzaka}
\affiliation{Center for Theoretical Physics, Polish Academy of Sciences, Aleja Lotnik\'{o}w 32/46, 02-668 Warsaw, Poland}

\author{Remigiusz Augusiak}
\affiliation{Center for Theoretical Physics, Polish Academy of Sciences, Aleja Lotnik\'{o}w 32/46, 02-668 Warsaw, Poland}

\begin{abstract}

Genuinely entangled subspaces are a class of subspaces in the multipartite Hilbert spaces that are composed of only genuinely entangled states. They are thus an interesting object of study in the context of multipartite
entanglement. Here we provide a construction of multipartite subspaces that are not only genuinely entangled but also fully non-positive-partial-transpose (NPT) in the sense that any mixed state supported on them has non-positive partial transpose across any bipartition. Our construction originates from the stabilizer formalism known for its use in quantum error correction. To this end, we first introduce a couple of criteria allowing to assess whether any state from a given non-trivial stabilizer subspace is genuinely multipartite entangled. We then use these
criteria to construct genuinely entangled stabilizer subspaces for any number of parties and prime local dimension and conjecture them to be of maximal dimension achievable within the stabilizer formalism. At the same time, we prove that every genuinely entangled subspace is fully NPT in the above sense, which implies
a quite surprising fact that no genuinely entangled stabilizer subspace can support PPT entangled states. 
\end{abstract}

\maketitle

\section{Introduction}

Entanglement is one of the most exhilarating features of quantum systems \cite{RevModPhys.81.865}. Its almost magic-like properties defy our real-life intuitions and show that the laws of the physical at a small scale cannot be reliably modelled with methods of classical physics. Throughout the years, many have tried, and succeeded, to use entanglement as a resource for seemingly infinite number of interesting applications that often have no classical analogue, such as the well-known quantum teleportation \cite{PhysRevLett.70.1895} or quantum cryptography \cite{Ekert91}.
Entanglement is also vital for quantum steering \cite{PhysRevLett.98.140402} or Bell nonlocality \cite{Bell} which represent other forms of quantum correlations that have been turned into independent resources themselves (see, e.g., \cite{Steeringreview,Bellreview14}). In particular, Bell nonlocality lies at
the heart of quantum information in the device-independent version \cite{PhysRevLett.98.230501}. It is thus not surprising that until now entanglement theory has been a very lively field of research.

Obviously, in order to study entanglement, one has to identify which states are entangled in the first place. The separability problem, that is, the problem of deciding whether a given, in general mixed, quantum state is entangled, has been thus intensively studied over years \cite{GUHNE20091}. At first, this problem was considered mostly for bipartite quantum states \cite{GURVITS2004448}, but it quickly gained on importance in the many-body regime. This was driven by the development of experimental techniques that allow one to prepare and control interesting many-body states (see, e.g., Refs. \cite{JDScience,Kasevich,Experimental}) but also by the identification of certain applications of multipartite entanglement such as quantum metrology \cite{QuantumMetro11} or quantum computing \cite{QuantumComp01}. Entanglement turned also useful in the study of many-body phenomena \cite{RevModPhys.80.517}.

Among many forms of entanglement featured by the multipartite scenario, the strongest and at the same time most valuable from the application point of view, is arguably the genuine multipartite entanglement \cite{G_hne_2010,PhysRevLett.106.190502,PhysRevA.83.040301}. Roughly speaking, a genuinely entangled multipartite quantum state 
is one that cannot be represented as a convex mixture of other states that
are separable with respect to some bipartitions. A significant amount of attention has been devoted in the literature to characterise the properties of genuine entanglement and to provide methods of its detection \cite{PhysRevLett.106.190502,PhysRevA.98.062102,Micuda:19,Eltschka2020maximumnbody,PhysRevLett.128.080507}. Yet, the problem is tremendously difficult, certainly more difficult than in the bipartite case, and our understanding of genuine entanglement still remains incomplete.

Our aim in this work is to join the effort of characterisation of genuine entanglement in multipartite systems, taking however, a slightly different perspective. Instead of focusing on particular multipartite quantum states (in general mixed) we rather consider entangled subspaces of multipartite Hilbert spaces. In this way we gain generality: a statement made for a subspace applies to any mixed state defined on it. While genuinely entangled subspaces have recently been an object of intensive exploration \cite{PhysRevA.98.012313,2020,PhysRevA.99.032335,Demianowicz_2021,demianowicz2021universal,Antipin_2021}, there are actually not many schemes allowing to judge whether a given subspace (and thus all states acting on it) is genuinely entangled. Here we will address this problem, concentrating on a particular class of multipartite subspaces that originate from the multiqudit stabilizer formalism \cite{GottesmanThesis}. The latter, being known for its use in quantum error correction \cite{PhysRevLett.77.793,KITAEV20032,PhysRevA.103.042420,Huber2020quantumcodesof}, provides an easy-to-handle and convenient description of a certain class of quantum states.  

We first provide a simple necessary and sufficient criterion allowing to
decide whether a given stabilizer subspace in an $N$-qudit
Hilbert space is genuinely entangled. Importantly, our criterion
reduces to checking certain commutation relations for operators
generating a given stabilizer subspace which can be done in a finite number of steps,
thus allowing us to avoid the tedious task of assessing whether every state belonging to that subspace is genuinely entangled. We then generalize the vector formalism introduced for the multiqubit case in Ref. \cite{Makuta_2021} that provides an efficient description of entanglement properties of stabilizer subspaces. This formalism is also particularly useful in constructing stabilizer subspaces that are genuinely entangled, and thus we employ it to provide a family of such subspaces which we believe to have the maximal dimension achievable within the stabilizer formalism. At the same time, we explore the use of partial transposition in deciding whether stabilizer subspaces are genuinely entangled and show that any mixed state supported on such a subspace must be NPT with respect to any bipartition. We thus provide a family of (possibly) maximally-dimensional genuinely entangled subspaces in $N$-qudit Hilbert spaces that are fully NPT, thus extending the results of Ref. \cite{PhysRevA.87.064302} to the multipartite regime.

\section{Preliminaries}

This section serves as an introduction of commonly used terms within the field of quantum information. Readers already familiar with concepts of genuine entanglement, partial transpose and stabilizer formalism are advised to skip ahead to Section \ref{sec_ge_stab}.\\

\textit{(1) Genuine entanglement.} 
Let us consider a scenario where $N$ parties share a pure state $\ket{\psi}\in \mathcal{H}$. We can decompose the Hilbert space $\mathcal{H}$ in the following manner
\begin{equation}\label{eq_hilbert_tensor}
\mathcal{H}=\bigotimes_{i=1}^{N} \mathcal{H}_{i},
\end{equation}
where $\mathcal{H}_{i}$ is a Hilbert subspace associated with $i$'th party. 
Let us then consider a partition of the set of parties $I_N:=\{1,\ldots,N\}$
into two disjoint subsets $Q$ and $\overline{Q}$ and call it a bipartition of $I_N$. We say that $\ket{\psi}$ is separable with respect the bipartition $Q|\overline{Q}$ if
\begin{equation}
\ket{\psi}= \ket{\psi}_{Q}\otimes \ket{\psi}_{\overline{Q}},
\end{equation}
for some $\ket{\psi}_{Q}\in \mathcal{H}'=\bigotimes_{i\in Q}\mathcal{H}_{i}$. 
Otherwise, we call it entangled with respect to the bipartition 
$Q|\overline{Q}$. Then, we call $\ket{\psi}$ genuinely multipartite entangled
iff it is entangled with respect to any nontrivial bipartition of 
the set $I_N$ \cite{GUHNE20091}.
Of course, no state could genuinely entangled if we were to consider trivial bipartitions, i.e., ones for which $Q= \emptyset$ or $Q=\{1,\dots,N\}$. 
Thus, whenever we say "all bipartitions" we only mean the nontrivial ones.

One can also introduce the above notions in the mixed-state case. Precisely, we say that a density matrix $\rho$ is entangled with respect to a given bipartition $Q|\overline{Q}$ if it cannot be represented as a convex combination of pure states that are separable across this bipartition. Then, $\rho$ is called genuinely multipartite entangled if it cannot be represented as a convex combination of pure state in which each pure state is separable across 
possibly different bipartitions.

In an analogous manner we can also extend these definitions to entire subspaces of the joint Hilbert space $\mathcal{H}$: a subspace is entangled with respect to some bipartition or it is genuinely entangled if every state from that subspace is of the corresponding type \cite{PhysRevA.98.012313}. In particular, a genuinely entangled subspace of $\mathcal{H}$ is one that consists of only genuinely entangled pure states.

\textit{(2) Positive partial transpose states.}
Consider a density matrix $\rho$ acting on $\mathcal{H}$. We call it a  non-positive-transpose (NPT) state with respect to a bipartition $Q|\overline{Q}$ if it its partial transposition with respect to $Q|\overline{Q}$ is not a positive semi-definite matrix,
\begin{equation}\label{eq_npt}
\rho^{T_{Q}}\ngeqslant 0,
\end{equation}
where $T_{Q}$ denotes a transposition over subsystems from $Q$ defined as
\begin{equation}
\rho^{T_{Q}}=\left(\bigotimes_{i\in Q} T^{(i)}\otimes \bigotimes_{j\in\overline{Q}}I^{(j)}\right)[\rho],
\end{equation}
where $T^{(i)}$ is the transposition in the standard basis performed on the $i$th subsystem, whereas $I^{(j)}$ is the identity on the subsystem $j$.
In other words, partial transpose is a transpose 
operation carried out over only some parties. As an example, let us consider a matrix 
\begin{equation}
C=\sum_{i,j,k,l}\gamma_{i,j,k,l}\ket{i}_{1}\bra{j}\otimes \ket{k}_{2}\bra{l},
\end{equation}
where $1,2$ indice indicate the party. Then the 
partial transpose over the first party of this 
matrix is defined as follows
\begin{equation}
C^{T_{1}}=\sum_{i,j,k,l}\gamma_{i,j,k,l}\ket{j}_{1}\bra{i}\otimes \ket{k}_{2}\bra{l}.
\end{equation}

Let us finally define fully NPT states to be those for which 
\eqref{eq_npt} holds true for any bipartition $Q|\overline{Q}$
We can also extend the above definitions of to subspaces of $\mathcal{H}$.
Precisely, a subspace $V\subset \mathcal{H}$ is termed NPT 
with respect to a bipartition $Q|\overline{Q}$ if every 
mixed state supported on it is NPT with respect to $Q|\overline{Q}$.
Analogously, $V$ is called fully NPT if any mixed state supported on it is fully NPT.

\textit{(3) Qudit stabilizer formalism.} 
In this work we concentrate on a class of quantum states (pure or mixed) that originate from the multiqudit stabilizer formalism. Let us assume that $N$ parties share a pure state $\ket{\psi}\in \mathcal{H}$. The decomposition of Hilbert space (\ref{eq_hilbert_tensor}) also holds true in this scenario, however in this case we additionally assume that for every $i\in\{1,\dots,N\}$ we have $\mathcal{H}_{i}=\mathbb{C}^{d}$, where $d$ is called a dimension of a local Hilbert space, or in other words each party holds one qudit. 

The most important objects in the qudit stabilizer formalism are 
the generalised Pauli operators defined as
\begin{equation}\label{eq_xz_def}
X=\sum_{i=0}^{d-1}\ket{i+1}\!\bra{i},\qquad Z=\sum_{i=0}^{d-1}\omega^{i}\ket{i}\!\bra{i},
\end{equation}
where $\omega=\operatorname{exp}(2\pi \mathbb{i}/d)$ is a primitive $d$th root of unity, $\mathbb{i}$ denotes the imaginary unit, and the addition is modulo $d$, meaning that $\ket{d}\equiv\ket{0}$. With a bit of calculation one can verify that these matrices enjoy the following properties:
\begin{equation}
X^{\dagger}X=\mathbb{1},\qquad Z^{\dagger}Z=\mathbb{1},
\end{equation}
\begin{equation}\label{eq_xz^d}
X^{d}=\mathbb{1},\qquad Z^{d}=\mathbb{1},
\end{equation}
meaning that they are unitary and their spectra are 
$\{1,\omega,\ldots,\omega^{d-1}\}$. Moreover, they obey 
the following commutation relations
\begin{equation}\label{eq_xz_zx}
X^{i}Z^{j}=\omega^{-ij}Z^{j}X^{i}
\end{equation}
for any $i,j=0,\ldots,d-1$.

Let us now introduce the generalised Pauli group $\mathbbm{P}_{N,d}$ as a set of all $N$-fold tensor products of matrices $\omega^{r} X^{n}Z^{m}$, where 
and $r,n,m\in \{0,\dots,d-1\}$, equipped with a matrix multiplication as the group operation. 

Consider then a subgroup $\mathbb{S}\subset\mathbbm{P}_{N,d}$. 
It is called a stabilizer if \cite{PhysRevA.71.042315}:
\begin{enumerate}
    \item all elements of $\mathbb{S}$ mutually commute,
    \item given $a\in\mathbb{C}$, we have $a\mathbb{1}\in\mathbb{S}$ if and only if $a=1$.
\end{enumerate} 
We say that a state $\ket{\psi}$ is stabilized by $\mathbb{S}$ if
\begin{equation}\label{eq_stab}
    G\ket{\psi}=\ket{\psi}
\end{equation}
holds true for any $G\in\mathbb{S}$. Then, a stabilizer subspace $V$ is a maximal subspace stabilized by a given stabilizer, i.e., a state $\ket{\psi}$ is stabilized by $\mathbb{S}$ if and only if $\ket{\psi}\in V$.

We can also extend the definition of stabilisation to the mixed states: we say that a state described by a density matrix $\rho$ is stabilized by $\mathbb{S}$, if for every operator $G\in \mathbb{S}$ we have 
\begin{equation}
G\rho=\rho G=\rho.
\end{equation}
For a more in depth discussion of a stabilizer formalism see \cite{GHEORGHIU2014505, nielsen00}.

Given a stabilizer $\mathbb{S}$, it is often convenient to represent it in terms of the smallest set of independent elements of $\mathbb{S}$, called generators, using which one can represent any other element of $\mathbb{S}$, where 'independent' means that one generator cannot be represented as a product of the other generators.
\begin{equation}
G_{i}\neq G_{1}^{\alpha_{1}}\dots G_{i-1}^{\alpha_{i-1}}G_{i+1}^{\alpha_{i+1}}\dots G_{k}^{\alpha_{k}}
\end{equation}
for $\alpha_{i}\in\{0,\dots,d-1\}$. Of course from the definition of a stabilizer it follows that:
\begin{equation}
    [G_i,G_j]=0\qquad (i\neq j).
\end{equation}
For prime $d$, this representation of a stabilizer allows us to easily determine the dimension of the corresponding stabilizer subspace $V$, namely $\dim V=d^{N-k}$; for non-prime $d$ the situation is slightly more complicated (see Ref. \cite{GHEORGHIU2014505}).

\section{Genuinely entangled subspaces in a qudit stabilizer formalism}
\label{sec_ge_stab}

In this section we aim to characterise genuine entanglement in a qudit stabilizer formalism and find a genuinely entangled subspace that is as large as possible. To this end, we generalise the results from \cite{Makuta_2021} that concern genuine entanglement in a stabilizer formalism for a local Hilbert space of dimension $d=2$.

Let us consider a nontrivial bipartition $Q|\overline{Q}$ of $\{1,\dots,N\}$ and a stabilizer $\mathbb{S}=\langle G_{1},\dots,G_{k}\rangle$ with a corresponding stabilizer subspace $V\subset \mathcal{H}=\mathbb{C}_{d}^{\otimes N}$. Each generator $G_{i}$ can be decomposed with respect to the above bipartition in the following way
\begin{equation}\label{eq_g_bip}
G_{i}=G_{i}^{(Q)}\otimes G_{i}^{(\overline{Q})},
\end{equation}
where $G_{i}^{(Q)}$, $G_{i}^{(\overline{Q})}$ act on $\bigotimes_{i\in Q}\mathcal{H}_{i}$, $\bigotimes_{i\in \overline{Q}}\mathcal{H}_{i}$, respectively. With this we can formulate the following theorem.
\begin{thm}\label{thm_sep}
Given a stabilizer $\mathbb{S}=\langle G_{1},\dots, G_{k}\rangle$ of an arbitrary local dimension, the corresponding stabilizer subspace $V\neq\{0\}$ is entangled with respect to a bipartition $Q|\overline{Q}$ iff there exists a pair $i,j\in \{1,\dots,k\}$ for which the following holds true:
\begin{equation}\label{eq_gigj_com}
\left[G_{i}^{(Q)},G_{j}^{(Q)}\right]\neq 0.
\end{equation}
\end{thm}
\begin{proof}
We will begin by proving the "$\Rightarrow$" implication by contradiction. Let us assume that $V$ is entangled with respect to a bipartition $Q|\overline{Q}$ and that for all $i,j\in\{1,\dots,N\}$ we have
\begin{equation}
\left[G_{i}^{(Q)},G_{j}^{(Q)}\right]=0.
\end{equation}
Since all generators commute, this implies that
\begin{equation}
\left[G_{i}^{(\overline{Q})},G_{j}^{(\overline{Q})}\right]=0
\end{equation}
for all  $i,j\in\{1,\dots,N\}$.

Consequently, there exists a common eigenbasis for all $G_{i}^{(Q)}$ and another common eigenbasis for all $G_{i}^{(\overline{Q})}$:
\begin{equation}\label{eq_g_bip_basis}
G_{i}^{(Q)}\ket{\phi_{j}}=\lambda_{j}^{(i)}\ket{\phi_{j}}\qquad \textrm{and} \qquad G_{i}^{(\overline{Q})}\ket{\theta_{j}}=\overline{\lambda}_{j}^{(i)}\ket{\theta_{j}}
\end{equation}
for $\lambda_{j}^{(i)}, \overline{\lambda}_{j}^{(i)}\in \mathbb{C}$. Using this fact, we can construct a basis of the joint Hilbert space $\mathcal{H}$ by taking the tensor product of elements of both bases, that is,
\begin{equation}
\ket{\varphi_{ij}}=\ket{\phi_{i}}\otimes\ket{\theta_{j}}.
\end{equation}
Importantly, it follows from Eqs. (\ref{eq_g_bip}) and (\ref{eq_g_bip_basis})  that this is a common eigenbasis of generators $G_{i}$, hence some $\ket{\varphi_{ij}}$ correspond to the eigenvalue one of all $G_{i}$, or $V$ is empty. The later contradicts the assumption that $V$ is non-empty, while the former implies that $\ket{\varphi_{ij}}$ belongs to $V$. Since $\ket{\varphi_{ij}}$ is clearly separable with respect to the bipartition $Q|\overline{Q}$, this contradicts the assumption that $V$ is entangled across this bipartition.

From the definition of Pauli group $\mathbb{P}_{N,d}$, $G_{i}$ is a tensor product of matrices $\omega^{r}X^{n}Z^{m}$ for $r,n,m\in\{0,\dots,d-1\}$.

Let us now move on to the "$\Leftarrow$" implication. Let us assume that $V$ is separable with respect to a bipartition $Q|\overline{Q}$ but there exists a pair of generators $G_{i},G_{j}$ for which \eqref{eq_gigj_com} holds true. This means that there exists $\ket{\psi}\in V$ such that for the bipartition $Q|\overline{Q}$ we have
\begin{equation}
\ket{\psi}=\ket{\psi}_{Q}\otimes \ket{\psi}_{\overline{Q}}.
\end{equation}
This together with (\ref{eq_g_bip}) implies that the following holds true
\begin{equation}
G_{i}^{(Q)}\ket{\psi}_{Q}=e^{\mathbb{i} \phi_{i}}\ket{\psi}_{Q}
\end{equation}
for all $i$, where $\phi_{i}\in\mathbb{R}$. This formula allows us to write
\begin{equation}\label{eq_g_anticom_1}
\left\{G_{i}^{(Q)},G_{j}^{(Q)}\right\}\ket{\psi}_{Q}=2e^{\mathbb{i}(\phi_{i}+\phi_{j})}\ket{\psi}_{Q},
\end{equation}
where $\{\cdot,\cdot\}$ is an anticommutator. On the other hand, from the definition of Pauli group $\mathbb{P}_{N,d}$, $G_{i}$ is a tensor product of matrices $\omega^{r}X^{n}Z^{m}$ for $r,n,m\in\{0,\dots,d-1\}$, which implies 
\begin{eqnarray}\label{eq_g_anticom_2}
\left\{G_{i}^{(Q)},G_{j}^{(Q)}\right\}\ket{\psi}_{Q}&=&\left(1+\omega^{\tau_{i,j}^{(Q)}}\right)G_{j}^{(Q)}G_{i}^{(Q)}\ket{\psi}_{Q}\nonumber\\&=&\left(1+\omega^{\tau_{i,j}^{(Q)}}\right)\mathrm{e}^{\mathbb{i}(\phi_i+\phi_j)}\ket{\psi}_{Q},\nonumber\\
\end{eqnarray}
where the value of $\tau_{i,j}^{(Q)}$ depends on the exact form of $G_{i}^{(Q)}$ and $G_{j}^{(Q)}$. After comparing Eqs. (\ref{eq_g_anticom_1}) and (\ref{eq_g_anticom_2}), one arrives at 
\begin{equation}
1+\omega^{\tau_{i,j}^{(Q)}}=2,
\end{equation}
which is fulfilled only if $\tau_{i,j}^{(Q)}=0$, which is to say that for the bipartition $Q|\overline{Q}$, for every $i,j\in\{1,\dots,k\}$ we have
\begin{equation}
\left[G_{i}^{(Q)},G_{j}^{(Q)}\right]=0,
\end{equation}
which contradicts the assumption.
\end{proof}

This theorem provides an alternative technique of determining whether a given subspace originating from the qudit stabilizer formalism is entangled. This  requires far less effort than checking the same property on the explicit form of the stabilizer subspace, in particular for every pure state belonging to it. Moreover, by applying Theorem \ref{thm_sep} to all nontrivial bipartitions we arrive at the following corollary.
\begin{cor}\label{cor_ge}
Given a stabilizer $\mathbb{S}=\langle G_{1},\dots, G_{k}\rangle$ of an arbitrary local dimension, the corresponding stabilizer subspace $V\neq\{0\}$ is genuinely entangled iff for each bipartition $Q|\overline{Q}$ there exists a pair $i,j\in \{1,\dots,k\}$ for which the following holds true:
\begin{equation}
\left[G_{i}^{(Q)},G_{j}^{(Q)}\right]\neq 0.
\end{equation}
\end{cor}

Let us notice that the above fact is a direct generalisation of Theorem 1 from Ref. \cite{Makuta_2021} and it allows us to check whether a stabilizer subspace is genuinely entangled just by examining generators of the corresponding stabilizer.

Having the above description of genuine entanglement within the qudit stabilizer formalism at hand, we now aim to find a genuinely entangled stabilizer subspace of the largest possible dimension. Luckily, we can once again refer to \cite{Makuta_2021} in which the maximal dimension of a genuinely entangled, stabilizer subspace for a case of $d=2$ was found. It has been achieved via a certain novel vector formalism, and so our first step in finding the aforementioned largest subspace is to generalise this formalism to arbitrary $d$. 

In the remainder of this section we assume that $d$ is a prime number. Let us consider the set $\mathbb{Z}_{d}=\{0,\dots,d-1\}$ equipped with addition and multiplication modulo $d$. We denote by $F_{N,d}=\mathbb{Z}_{d}^{N}$ a vector space over $\mathbb{Z}_{d}$, i.e., every $f\in F_{N,d}$ is a $N$-dimensional vector whose entries are from $\{0,\dots,d-1\}$. 

We can take advantage of this vector space to find a representation of a stabilizer $\mathbb{S}=\langle G_{1},\dots,G_{k}\rangle$ that gives us a clearer view on the entanglement of the subspace $V$. Let us consider a vector
\begin{equation}\label{eq_vij_def}
v_{i,j}=\sum_{n=1}^{N}\tau_{i,j}^{(n)}e_{n},
\end{equation}
where $e_{n}$ is a unit vector that has an entry $1$ on the $n$'th side and $0$ elsewhere, and $\tau_{i,j}^{(n)}\in\{0,\dots,d-1\}$ is defined as follows
\begin{equation}\label{eq_g_com}
G_{i}^{(n)}G_{j}^{(n)}=\omega^{\tau_{i,j}^{(n)}}G_{j}^{(n)}G_{i}^{(n)}.
\end{equation}
For the definition of a stabilizer, all generators $G_{1},\dots, G_{k}$ have to mutually commute, which implies that for all $i$, $j$ we have
\begin{equation}\label{eq_vij_com}
\sum_{n=1}^{N}v_{i,j}^{(n)}=0,
\end{equation}
where $v_{i,j}^{(n)}$ is an $n$'th entry of $v_{i,j}$. Let us denote by $K(\mathbb{S})$ a subspace of $F_{N,d}$ that is spanned by all the vectors $v_{i,j}$ (\ref{eq_vij_def}) corresponding to the stabilizer $\mathbb{S}$. It is not difficult to see that for any stabilizer $\mathbb{S}$, $K(\mathbb{S})$ is a proper subspace of $F_{N,d}$: there exist vectors in $F_{N,d}$ that do not satisfy the condition (\ref{eq_vij_com}), and, at the same time, any linear combination of vectors that fulfil (\ref{eq_vij_com}) fulfils (\ref{eq_vij_com}) too. In fact, 
the condition (\ref{eq_vij_com}) implies that the maximal dimension of $K(\mathbb{S})$ is $N-1$.

To make full use of the above formalism we need to reformulate Corollary \ref{cor_ge} in terms of vectors $v_{i,j}$. To this end, we need to define a representation of a bipartition of the set $\{1,\dots,N\}$, and so let us consider a vector $\phi\in F_{N,2}$ (notice that this is $F_{N,2}$ and not $F_{N,d}$). We call $\phi$ a representation of a bipartition $Q|\overline{Q}$ if the following holds true 
\begin{equation}
\phi^{(n)}=1 \Longleftrightarrow n\in Q.
\end{equation}
Notably, this is not an isomorphism, since every $Q|\overline{Q}$ has two equivalent representation: one where $\phi^{(n)}=1$ for $n\in Q$ and the other one for $n\in \overline{Q}$. 

This representation also includes the trivial bipartitions, and since we usually want to exclude them, let us denote representations of the trivial bipartition as
\begin{equation}
\phi_{T_{0}}=\sum_{n=1}^{N}0 e_{n},\qquad \phi_{T_{1}}=\sum_{n=1}^{N} e_{n},
\end{equation}
where $T_0$ and $T_1$ correspond to the cases $Q=\emptyset$ and $\overline{Q}=\emptyset$, respectively.

The last ingredient needed to rephrase Corollary \ref{cor_ge} in this vector formalism is a function $h(\cdot,\cdot):$ $F_{N,d}\times F_{N,2}\rightarrow \mathbb{Z}_{d}$:
\begin{equation}\label{eq_h_def}
h(v,\phi)=\sum_{n=1}^{N}v^{(n)}\phi^{(n)}.
\end{equation}
This function resembles a scalar product in the sense that it is calculated similarly, however $h(\cdot,\cdot)$ does not meet the necessary conditions that a scalar product has to obey.

We are finally ready to reformulate Corollary \ref{cor_ge} in terms of vectors from $K(\mathbb{S})$.
\begin{lem}\label{lem_ge}
Let the local dimension $d$ be a prime number. Given a stabilizer $\mathbb{S}=\langle G_{1},\dots, G_{k}\rangle$, the corresponding stabilizer subspace $V\neq\{0\}$ is genuinely entangled iff for every $\phi\in F_{N,2}$ such that $\phi\neq \phi_{T_{0}},\phi_{T_{1}}$ there exists $v_{i,j}\in K(\mathbb{S})$ for $i\neq j$ such that
\begin{equation}
h(v_{i,j},\phi)\neq 0.
\end{equation}
\end{lem}
The proof is very similar to the proof of Theorem \ref{thm_sep} and so we do not present it here. However, for completeness we provide it in Appendix \ref{app_lem}.

The above lemma gives us a clearer picture of the relations that generators $G_{i}$ have to fulfil in order for a stabilizer subspace to be genuinely entangled. But notably in Ref. \cite{Makuta_2021} this lemma only serves as a stepping stone to a proof of the following theorem.
\begin{thm}[\textbf{for} $\boldsymbol{d=2}$]\label{thm_d=2}
A subspace $V$ stabilized by $\mathbb{S}$
is genuinely entangled iff $\dim K(\mathbb{S})=N-1$.
\end{thm}
This theorem provides a more convenient criterion than Corollary \ref{cor_ge} for determining if a stabilizer subspace is genuinely entangled. In fact, 
checking whether a stabilizer subspace generated by a stabilizer $\mathbb{S}$
is genuinely entangled boils down to computing the dimension of the corresponding subspace $K(\mathbb{S})$ which can easily be implemented in standard mathematical packages. However, the above theorem does not directly generalize to the case of arbitrary prime $d>2$. More precisely, while the "$\Leftarrow$" implication remains true, the other one does not hold anymore because there exist genuinely entangled subspaces for which the dimension of the corresponding subspaces $K(\mathbb{S})$ is lower than $N-1$. In what follows we provide an example of such a subspace. We also prove the sufficient condition for a stabilizer subspace to be genuinely entangled for any $d$, generalizing the "if part" of Theorem \ref{thm_d=2} to stabilizer subspace of arbitrary local dimension $d$.

To demonstrate that Theorem \ref{thm_d=2} cannot be directly generalized to an arbitrary $d$ let us consider the following example of stabilizers for $d=3$ and $N=3$:
\begin{equation}
G_{1}=X\otimes X\otimes X, \quad G_{2}=Z\otimes Z\otimes Z,
\end{equation}
\begin{equation}
G'_{1}=X\otimes Z\otimes \mathbb{1},\quad G'_{2}=Z\otimes X\otimes \mathbb{1}.
\end{equation}
It is easy to see that the corresponding stabilizers $\mathbb{S}=\langle G_{1},G_{2}\rangle$ and $\mathbb{S}'=\langle G'_{1}, G'_{2}\rangle$ stabilise three-dimensional subspaces in $\mathbb{C}_3^{\otimes 3}$, which we denote respectively by $V$ and $V'$. Using Corollary \ref{cor_ge} we can show that $V$ is genuinely entangled whereas $V'$ is not. At the same time, $\dim K(\mathbb{S})=\dim K(\mathbb{S}')=1$. This is thus certainly a nail in the coffin 
for the possibility of a direct generalisation of Theorem \ref{thm_d=2} to higher dimensional cases because it is clear that $\dim K(\mathbb{S})$ alone does not determine in general whether $V$ is genuinely entangled as it does for $d=2$.

While it clearly follows from the above considerations that 
it is in general not possible to formulate an iff condition for a stabilizer subspace to be genuinely entangled based solely on the dimension of $K(\mathbb{S})$ (as it is the case for $d=2$), we believe that the latter can still be somewhat useful in assessing genuine entanglement of stabilizer subspaces. While the latter condition implies in particular that $G_i^d=\mathbbm{1}$ for every generator, Actually, we conjecture that the following necessary condition holds true.
\begin{con}\label{con_ge}
Given a stabilizer $\mathbb{S}=\langle G_{1},\dots,G_{k}\rangle$, if the local dimension $d$ is a prime number and the corresponding stabilizer subspace $V$ is genuinely entangled, then
\begin{equation}\label{eq_dim_ks}
\dim K(\mathbb{S}) \geqslant \left\lceil \frac{N-1}{d-1} \right\rceil.
\end{equation}
\end{con}
For the particular case of $k=2$ and any $N$ and $d$ the above statement can actually be rigorously proven. Precisely, the following theorem holds true. 
\begin{thm}\label{Raimat}
Consider a stabilizer $\mathbb{S}=\langle G_1,G_2\rangle$. If the local dimension $d$ is a prime number and the corresponding 
subspace $V$ generated by $\mathbb{S}$ is genuinely entangled, then 
\begin{equation}\label{eq_dimks}
    \dim K(\mathbb{S})\geqslant \left\lceil \frac{N-1}{d-1} \right\rceil.
\end{equation}
\end{thm}
\begin{proof}
The proof is quite technical and therefore it is deferred to Appendix \ref{ap_con}.
\end{proof}

Furthermore, in Appendix \ref{ap_gmax} we construct a stabilizer for which $V$ is genuinely entangled and $\dim K(\mathbb{S})=\lceil (N-1)/(d-1)\rceil$; we denote this stabilizer $\mathbb{S}_{\max}$, where we added the subscript '$\max$' for reasons that will become clear later. This implies that if the inequality (\ref{eq_dim_ks}) is in general false, the true function $p(N,d)$ that bounds $\dim K(\mathbb{S})$ from below, that is, one for which $\dim K(\mathbb{S})\geqslant  p(N,d)$ for any genuinely entangled subspace, must satisfy $p(N,d)\leqslant \lceil (N-1)/(d-1)\rceil$ for any $N$ and $d$.

It is worth pointing out the vector formalism introduced above turned out to be particularly useful in constructing the stabilizer $\mathbb{S}_{\max}$ presented in Appendix \ref{ap_gmax}. Precisely, the generators of this stabilizer are chosen so that the corresponding vectors $v_{i,j}$ give us the following basis in $K(\mathbb{S})$:
\begin{equation}\label{eq_basis}
u_{i}=\left\{ 
\begin{array}{ll}
\displaystyle\sum_{j=(i-1)(d-1)+1}^{i(d-1)+1}e_{j}& \textrm{ for\;} i < \displaystyle\left\lceil \frac{N-1}{d-1} \right\rceil, \\[5ex]
\displaystyle\sum_{j=(i-1)(d-1)+1}^{N-1}e_{j}+ a e_{N}& \textrm{ for\;}i = \displaystyle\left\lceil \frac{N-1}{d-1} \right\rceil,
\end{array}
\right.
\end{equation}
where $a=\lceil (N-1)/(d-1) \rceil(d-1)-(N-1)+1$. For example, for $d=3$ and $N=5$ our construction gives
\begin{equation}
u_{1}=(1,1,1,0,0),\qquad u_{2}=(0,0,1,1,1).
\end{equation}

To complete our characterization of genuine entanglement within the stabilizer formalism let us finally demonstrate that the '$\Leftarrow$' implication of Theorem \ref{thm_d=2} generalizes to any prime $d$. 
To this end let us first prove the following lemma.
\begin{lem}\label{lem_dim}
Assume $d$ to be prime and consider a stabilizer for which $\dim K(\mathbb{S})=N-1$. Then, a vector $v\in F_{N,d}$ fulfils
\begin{equation}\label{eq_sum_v_0}
\sum_{i=1}^{N}v^{(i)}=0,
\end{equation}
if and only if $v\in K(\mathbb{S})$.
\end{lem}
The proof can be found in Appendix \ref{app_lem}. With this lemma at hand we can prove the aforementioned implication. 
\begin{thm}
Given a prime $d$ and a stabilizer $\mathbb{S}=\langle G_{1},\dots, G_{k}\rangle$, if $\dim K(\mathbb{S})=N-1$ then the corresponding stabilizer subspace $V$ is genuinely entangled. 
\end{thm}
\begin{proof}
To prove that the stabilizer subspace in genuinely entangled we use the criterion from Lemma \ref{lem_ge}. Let us consider a representation $\phi$ of nontrivial bipartition $Q|\overline{Q}$. Without any loss of generality we can assume that $n_{0}\in \overline{Q}$ and $n_{1}\in Q$. Consider then a vector 
$v=(d-1)e_{n_{0}}+e_{n_{1}}$ for which we have $h(v,\phi)\neq 0$. 
Moreover, this vector clearly satisfies the condition (\ref{eq_sum_v_0})
which together with the fact that $\dim K(\mathbb{S})=N-1$
allows one to deduce \textit{via} Lemma \ref{lem_dim} that $v\in K(\mathbb{S})$. Now, such a vector can be found for every nontrivial $\phi$ and every such vector is in $K(\mathbb{S})$, therefore by the virtue of Lemma \ref{lem_ge} the stabilizer subspace $V$ is genuinely entangled. 
\end{proof}

Let us conclude this section by wrapping up the main statements made in this section and picturing how the fact whether $V$ is genuinely entangled or not depends on the dimension of $K(\mathbb{S})$. Assuming that $d$ is prime, we can divide the possible values of $\dim K(\mathbb{S})$ into three regions: 
\begin{itemize}
\item[(i)] for $\dim K(\mathbb{S})=N-1$ the corresponding stabilizer subspace is 
genuinely entangled, 
\item[(ii)]  for $\dim K(\mathbb{S})<N-1$ but larger or equal to some number $p(N,d)$, which we suspect to be $\lceil (N-1)/(d-1)\rceil$, we are unable to determine whether $V$ is genuinely entangled or not just by looking at the dimension of $K(\mathbb{S})$,
\item[(iii)]  for $\dim K(\mathbb{S})$ smaller that $p(N,d)$, $V$ cannot be genuinely entangled. However, since we were not able to prove Conjecture \ref{con_ge}, the actual value of $p(N,d)$ remains unknown.
\end{itemize}

\subsection{The maximal dimension of genuinely entangled stabilizer subspaces}

Interestingly, the stabilizer $\mathbb{S}_{\max}$ introduced in Appendix \ref{ap_gmax} and discussed above stabilises a genuinely entangled subspace in 
$\mathbb{C}_d^{\otimes N}$ whose dimension seems to be the maximal achievable within the stabilizer formalism for prime $d$. While the maximal dimension of a genuinely entangled stabilizer subspace for any $N$ and $d$
is unknown, we can speculate on it based on Conjecture \ref{con_ge}. Indeed, we can formulate another conjecture towards this goal. 
\begin{con}\label{con_dim}
If $d$ is prime then the maximal dimension of genuinely entangled, stabilizer subspace is given by $d^{N-k_{\min}(N)}$, where 
\begin{equation}\label{eq_k_min}
k_{\min}(N,d)=\left\lceil \frac{1+\sqrt{1+8\lceil (N-1)/(d-1)\rceil}}{2} \right\rceil.
\end{equation}
\end{con}
Let us now provide a proof that Conjecture \ref{con_ge} implies Conjecture \ref{con_dim}.
\begin{proof}Assume that Conjecture \ref{con_ge} is true. We first use the fact, proven in Ref. \cite{GHEORGHIU2014505}, that for a given stabilizer $\mathbb{S}=\langle G_{1},\dots,G_{k}\rangle$, the dimension of $V$ equals $d^{N-k}$ 
provided that $d$ is prime. Let us then follow the argumentation of 
the proof of Theorem 3 in Ref. \cite{Makuta_2021}. Namely, $k_{\min}(N,d)$ in Eq. (\ref{eq_k_min}) is the smallest positive integer that fulfils the following inequality 
\begin{equation}\label{eq_n-1/d-1}
\frac{1}{2}k(k-1)\geqslant \left\lceil \frac{N-1}{d-1} \right\rceil
\end{equation}
which relates the lower bound on the dimension of 
$K(\mathbb{S})$ in (\ref{eq_dim_ks}) and the number of 
vectors $v_{i,j}$ corresponding to the stabilizer 
$\mathbb{S}=\langle G_{1},\dots,G_{k}\rangle$. 
Indeed, for a stabilizer subspace $V$ to 
be genuinely entangled that number of vectors 
must at least equal the minimal dimension of $K(\mathbb{S})$
in Conjecture \ref{con_ge}.
\end{proof}

Since this conjecture is a direct consequence of Conjecture \ref{con_ge} it follows that the stabilizer constructed in Appendix \ref{ap_gmax} gives rise to a genuinely entangled subspace of the maximal dimension postulated in Conjecture \ref{con_dim}. Thus, the corresponding subspace $V_{\max}$ (see Appendix \ref{ap_gmax}) is the one that we were looking for. 

Let us also notice that if the true function $p(N,d)$ that bounds the dimension of $K(\mathbb{S})$ from below in (\ref{eq_dim_ks}) turns out to be different than 
$\lceil (N-1)/(d-1)\rceil$, the reasoning presented in the proof of Conjecture \ref{con_dim} remains valid, one only needs to replace the inequality (\ref{eq_n-1/d-1}) by $k(k-1)/2\geqslant p(N,d)$ and modify accordingly $k_{\min}(N,d)$ in Eq. (\ref{eq_k_min}).

\section{Fully NPT subspaces and genuine entanglement}
\label{NPT}

Our aim in the last section of the paper is to explore the relation between genuine multipartite entanglement in the stabilizer formalism and the entanglement criterion based on partial transposition. We show that partial transposition can be used to formulate an iff criterion for stabilizer subspaces to be genuinely entangled. Precisely, a mixed state supported on a stabilizer subspace is 
genuinely entangled if and only if all its partial transpositions are non-positive. 
This implies a quite unexpected fact that genuinely entangled, stabilizer subspaces, independently of how large they are, do not support mixed states whose partial transpositions are all positive. In order to achieve this goal we exploit the very simple approach used in Ref. \cite{PhysRevA.61.062102} to prove that no bipartite PPT entangled states violate the Clauser-Horne-Shimony-Holt Bell inequality. 

Before stating our results let us recall
that we say that a stabilizer subspace $V$ is entangled with respect to a bipartition
$Q|\overline{Q}$ if every pure state belonging to it
(and thus every density matrix supported on it)
is entangled across $Q|\overline{Q}$.
We then call $V$ a non-positive partial transpose (NPT) subspace with respect to $Q|\overline{Q}$ if $\rho^{T_Q}\ngeq 0$ for every mixed state $\rho$ supported on it (i.e., such that $\mathrm{supp}(\rho)\subseteq V$).

\begin{thm}\label{thm_npt}
Consider a stabilizer $\mathbb{S}=\langle G_{1},\dots, G_{k}\rangle$ of an arbitrary local dimension and a bipartition $Q|\overline{Q}$. The stabilizer subspace $V$ of $\mathbb{S}$ is entangled with respect to $Q|\overline{Q}$ iff it is NPT with respect to that bipartition. 
\end{thm}
\begin{proof}
The '$\Leftarrow$' implication simply follows from a well known fact that non-positive partial transpose is a sufficient condition for entanglement \cite{PhysRevLett.77.1413}, and so we only need to prove the "$\Rightarrow$" implication. Let us assume that $V$ is entangled with respect to the bipartition $Q|\overline{Q}$, but it is not NPT. This implies that there exists a state $\rho$, where $\operatorname{supp}(\rho)\subset V$, such that 
\begin{equation}\label{PPTQ}
\rho^{T_{Q}}\geqslant 0.
\end{equation}
The fact that $V$ is entangled across $Q|\overline{Q}$ implies, by virtue of Theorem \ref{thm_sep}, that there exists a pair of generators of $\mathbb{S}$, $G_{i}$ and $G_{j}$, such that:
\begin{equation}\label{eq_gq}
    G_{i}^{(Q)} G_{j}^{(Q)}=\omega^{\tau_{i,j}^{(Q)}}G_{j}^{(Q)}G_{i}^{(Q)},
\end{equation}
where $\tau_{i,j}^{(Q)}\neq 0$. Let us introduce an operator $\mathcal{B}$ as a sum of these two generators $\mathbb{S}$ and their Hermitian conjugates:
\begin{equation}\label{eq_B_def}
    \mathcal{B}= G_{i}+G_{i}^{\dagger}+G_{j}+G_{j}^{\dagger}.
\end{equation}
Now, following the approach Ref. \cite{PhysRevA.61.062102} we derive an upper bound on the expectation value of the operator $\mathcal{B}$ for any state that obeys (\ref{PPTQ}):
\begin{eqnarray}\label{eq_B_bound}
  [  \operatorname{tr}(\mathcal{B}\rho)]^{2}&=&\frac{1}{2}[\operatorname{tr}(\mathcal{B}\rho)]^{2}+\frac{1}{2}\left[\operatorname{tr}(\mathcal{B}^{T_{Q}}\rho^{T_{Q}})\right]^{2}\nonumber\\
    &\leqslant& \frac{1}{2}\operatorname{tr}(\mathcal{B}^{2}\rho)
    +\frac{1}{2}\operatorname{tr}(\mathcal{B}_{Q}^2\rho^{T_Q})\nonumber\\
    &=&\frac{1}{2}\operatorname{tr}\left\{\left[\mathcal{B}^{2}+\left(\mathcal{B}_{Q}^2\right)^{T_Q}\right]\rho\right\},
\end{eqnarray}
where $\mathcal{B}_{Q}=\mathcal{B}^{T_{Q}}$ and to obtain the first and the third equations we have used the fact that under the trace one can add partial transposition to both operators. Then, to obtain the inequality we exploited the fact that the variance of a quantum observable is non-negative, together with the fact that $\rho^{T_{Q}}\geqslant 0$. In order to determine the right-hand side of (\ref{eq_B_bound}) we use the explicit forms of the operators $\mathcal{B}^2$ and $\left(\mathcal{B}^2_{Q}\right)^{T_Q}$,
\begin{eqnarray}
\mathcal{B}^{2}&&=G_{i}^{2}+G_{i}^{\dagger 2}+G_{j}^{2}+G_{j}^{\dagger 2}+4\mathbb{1}+2g_{i,j},
\end{eqnarray}
\begin{eqnarray}
\left(\mathcal{B}^2_{Q}\right)^{T_Q}&&= G_{i}^{2}+ G_{i}^{2\dagger}+G_{j}^{2}+G_{j}^{2\dagger}+4\mathbb{1}+2\cos \left(\phi\right)g_{i,j},\nonumber\\
\end{eqnarray}
where $\phi=2\pi\tau_{i,j}^{(Q)}/d$ and
\begin{equation}
g_{i,j}=G_{i}G_{j}+G_{i}G_{j}^{\dagger}+G_{i}^{\dagger}G_{j}+G_{i}^{\dagger}G_{j}^{\dagger},
\end{equation}
and so we have:
\begin{eqnarray}
[\operatorname{tr}(\mathcal{B}\rho)]^{2}\leqslant&& \operatorname{tr}\Big[\rho\Big( G_{i}^{2}+G_{i}^{2\dagger}+G_{j}^{2}+G_{j}^{2\dagger}\nonumber\\
&&\hspace{1cm}+4\mathbb{1}+\left(1+\cos(\phi)\right)g_{i,j} \Big)\Big].
\end{eqnarray}
Then, taking into account the fact that 
\begin{equation}\label{Gcond}
    G_i\rho=\rho G_i
\end{equation}
holds true for any $i$,
the above simplifies to 
\begin{eqnarray}\label{contradiction}
\operatorname{tr}(\mathcal{B}\rho)&\leqslant& 2\sqrt{3+\cos(\phi)},\nonumber\\
&<&4,
\end{eqnarray}
where the last inequality is a consequence of the fact that $\tau_{i,j}^{(Q)}\neq 0$, which means that $\cos (\phi)<1$. Thus, for any mixed state stabilized on $V$ obeying (\ref{PPTQ}), the expectation value of $\mathcal{B}$ is strictly smaller than four. 

On the other hand, due to Eq. (\ref{Gcond}) every density matrix $\rho$ for which $\operatorname{supp}(\rho)\subset V$ satisfies
\begin{equation}
\operatorname{tr}(\mathcal{B}\rho)=4
\end{equation}
which clearly contradicts (\ref{contradiction}), completing the proof.
\end{proof}
This theorem shows that in the qudit stabilizer formalism subspaces in which 
all states are entangled with respect to some bipartition cannot support entangled states that are PPT with respect to that bipartition. Importantly, based on it one can formulate the following iff criterion for a stabilizer subspace to be genuinely entangled.
\begin{cor}\label{cor_npt}
A stabilizer subspace $V$ is genuinely entangled iff it is fully NPT.
\end{cor}
This statement follows from the fact that if $V$ is genuinely entangled then 
it is entangled with respect to any bipartition. Then, application of Theorem \ref{thm_npt} to every bipartition implies that $V$ is fully NPT. On the other hand, if every pure state from $V$ is fully NPT, it must also be genuinely entangled and thus $V$ is genuinely entangled itself.

This result provides us with a convenient way of constructing multipartite subspaces that are fully NPT: by virtue of Corollary \ref{cor_npt} any genuinely entangled stabilizer subspace does the job. At the same time, it implies a quite surprising fact that genuinely entangled stabilizer subspaces do not support genuinely entangled states that are PPT with respect to any bipartition, which makes us believe that there are actually no PPT genuinely entangled states in the stabilizer formalism. Let us notice, on the other hand, that while Theorem \ref{thm_npt} provides us a method for detecting entanglement in the stabilizer formalism, the latter role is, however, better served by Theorem \ref{thm_sep}.

\section{Conclusions}
We have shown that all genuinely entangled stabilizer subspaces are fully NPT, or, equivalently, that such subspaces do not support PPT entangled states. To this end, we have introduced a criterion allowing one to judge in a finite number of steps whether a given stabilizer subspace in an $N$-qudit Hilbert space is genuinely entangled. We have also generalized the vector formalism introduced
in Ref. \cite{Makuta_2021} in the case of qubit stabilizer formalism for prime local dimension and
used it to construct genuinely entangled stabilizer subspaces which we conjecture to have the maximal dimension achievable within the stabilizer formalism. We have thus provided examples of genuinely entangled fully NPT subspaces, extending in a way the results of Ref. \cite{PhysRevA.87.064302} to the multipartite scenario.

Our work leaves many open questions that might serve as inspiring starting points for further research. The most obvious one is whether the first conjecture stated in our paper can be proven rigorously. On the other hand, if this conjecture fails to be true, an alternative question would be whether one can find the true function $p(N,d)$ that should appear in Ineq. (\ref{eq_dimks}) and that can later be used to determine the maximal dimension of genuinely entangled stabilizer subspaces. Another related question is whether one can provide a general construction of fully NPT genuinely entangled subspaces beyond the stabilizer formalism that are of maximal dimension; it was show in Ref. \cite{PhysRevA.87.064302} that completely entangled subspaces that are NPT can be as large as any completely entangled subspaces. While we believe our construction to be of maximal dimension achievable within the stabilizer formalism, at least for prime $d$, it is most probably suboptimal when one takes into account arbitrary genuinely entangled subspaces. On the other hand, it would be highly interesting to follow the results of Sec. \ref{NPT} and further explore the relation between entanglement in the stabilizer formalism and the separability criterion based on partial transposition. While we prove that genuinely entangled stabilizer subspaces do not support PPT states, it is at the moment unclear whether there are no PPT genuinely entangled stabilizer states at all. In fact, there might exist stabilizer subspaces that are not genuinely entangled and yet support PPT genuinely entangled states. It would thus be highly interesting to explore whether there any PPT entangled states in this formalism, or, alternatively whether partial transposition is a necessary and sufficient condition for separability for stabilizer states. The last avenue for further research is to construct Bell inequalities that are maximally violated by genuinely entangled stabilizer subspaces and see whether they can be used for self-testing, generalizing the results Refs. \cite{PhysRevLett.125.260507,Makuta_2021}.

\section{acknowledgements}
This work was supported by the National Science Center (Poland) through the SONATA BIS grant no. 2019/34/E/ST2/00369.

\appendix
\section{Proofs of Lemmas \ref{lem_ge} and  \ref{lem_dim}}\label{app_lem}
\setcounter{lem}{0}

For completeness we provide here the proofs of 
Lemmas \ref{lem_ge} and \ref{lem_dim}.

\begin{lem}
Let the local dimension $d$ be prime. Given a stabilizer $\mathbb{S}=\langle G_{1},\dots, G_{k}\rangle$, the corresponding stabilizer subspace $V\neq\{0\}$ is genuinely entangled iff for every $\phi\in F_{N,2}$ such that $\phi\neq \phi_{T_{0}},\phi_{T_{1}}$ there exists $v_{i,j}\in K(\mathbb{S})$ for $i\neq j$ such that
\begin{equation}
h(v_{i,j},\phi)\neq 0.
\end{equation}
\end{lem}

\begin{proof}
Corollary \ref{cor_ge} states that for every nontrivial bipartition $Q|\overline{Q}$ of the set $\{1,\dots,N\}$ there exists $i,j\in \{1,\dots,k\}$ such that
\begin{equation}
\left[G_{i}^{(Q)},G_{j}^{(Q)}\right]\neq 0.
\end{equation}
This can be equivalently expressed as
\begin{equation}\label{eq_tau}
\tau_{i,j}^{(Q)}\neq 0,
\end{equation}
where
\begin{equation}
\tau_{i,j}^{(Q)}=\sum_{n\in Q}\tau_{i,j}^{(n)}
\end{equation}
and $\tau_{i,j}^{(n)}$ is defined in Eq. (\ref{eq_g_com}). Given the representation $\phi$ of a bipartition $Q|\overline{Q}$, the above equation can be reformulated as
\begin{equation}
\tau_{i,j}^{(Q)}=\sum_{n\in Q}\tau_{i,j}^{(n)}=\sum_{n=1}^{N}\tau_{i,j}^{(n)}\phi^{(n)}.
\end{equation}
By the definitions of $v_{i,j}$ [cf. Eq. \eqref{eq_vij_def}] and of $h(\cdot,\cdot)$ (\ref{eq_h_def}), the above rewrites as
\begin{equation}
\tau_{i,j}^{(Q)}=\sum_{n=1}^{N}\tau_{i,j}^{(n)}\phi^{(n)}=h(v_{i,j},\phi),
\end{equation}
Consequently, from \eqref{eq_tau} it follows that $h(v_{i,j},\phi)\neq 0$ for all nontrivial $\phi$ and some $v_{i,j}\in K(\mathbb{S})$ iff $V$ is genuinely entangled, which ends the proof.
\end{proof}

\begin{lem}
Assume $d$ to be prime and consider a stabilizer for which $\dim K(\mathbb{S})=N-1$. Then, a vector $v\in F_{N,d}$ fulfils
\begin{equation}
\sum_{i=1}^{N}v^{(i)}=0,
\end{equation}
if and only if $v\in K(\mathbb{S})$.
\end{lem}
\begin{proof}
Since the implication "$\Leftarrow$" follows from the definition of $K(\mathbb{S})$, in order to complete the proof it is sufficient to show that the number of vectors in $F_{N,d}$ fulfilling \eqref{eq_sum_v_0} equals the number of vectors in $K(\mathbb{S})$.

Let us start from calculating how many vectors in $F_{N,d}$ satisfy \eqref{eq_sum_v_0}. Clearly, for any such vector we have
\begin{equation}
 v^{(N)}=-\sum_{n=1}^{N-1}v^{(n)}.
\end{equation}
Notice that for any choice of $v^{(n)}$ for all $n\in\{1,\dots, N-1\}$ we can find exactly one $v^{(N)}$ that fulfils this equation. This implies that there are $d^{N-1}$ vectors $v\in F_{N,d}$ fulfilling \eqref{eq_sum_v_0}.

Next, let us calculate the number of vectors in $K(\mathbb{S})$. To this end let us consider the following expression
\begin{equation}
\sum_{i=1}^{N-1} a_{i}u_{i}= \sum_{i=1}^{N-1}b_{i}u_{i},
\end{equation}
where $\{u_{i}\}_{i=1}^{N-1}$ is a basis in $K(\mathbb{S})$ and $a_{i},b_{i}\in \mathbb{Z}_{d}$; notice that by definition each element of this basis must also satisfy (\ref{eq_sum_v_0}). Since $d$ is prime, the above is true only if $a_{i}=b_{i}$ for all $i\in \{1,\dots,N-1\}$. There are $d^{N-1}$ unique choices of $\{a_{i}\}_{i=1}^{N-1}$, hence there are $d^{N-1}$ vectors in $K(\mathbb{S})$, which ends the proof.
\end{proof}

\section{A proof of Theorem \ref{Raimat}}\label{ap_con}
Here we provide a proof of Theorem \ref{Raimat} or, equivalently, of Conjecture \ref{con_ge} for $k=2$ and arbitrary $d$ and $N$. If someone will attempt to prove Conjecture \ref{con_ge} in its most general form, we hope that the following proof will serve as guide for how the general proof could be formulated. The main idea of this proof is to treat bipartition representations $\phi$ as solutions to an equation:
\begin{equation}
h(u,\phi)=\beta,
\end{equation}
where $u\in F_{N,d}$ and $\beta\in\{0,\dots,d-1\}$. The first step in our procedure is to calculate the number of solutions $\phi$ to this equation depending on $d$ and $N$.
\begin{lem}\label{lem_2^(-(d-1))}
Let $u\in F_{N,d}$ and let us assume that an equation
\begin{equation}\label{eq_h_beta_simple}
h(u,\phi)=\beta,
\end{equation}
has at least one solution $\phi\in F_{N,2}$.
The number of $\phi$ that fulfil Eq. (\ref{eq_h_beta_simple}) is larger or equal than $2^{N-(d-1)}$.
\end{lem}
\begin{proof}
Let $\sigma (N,\beta)$ denote the number of $\phi\in F_{N,2}$ that fulfil:
\begin{equation}\label{eq_h_beta_2}
\sum_{i=1}^{N}u^{(i)}\phi^{(i)} \mod d=\beta,
\end{equation}
where the upper index denotes the $i$'th entry.

Rather than finding a lower bound on $\sigma(N,\beta)$ for each $\beta\in\{0,\dots,d-1\}$, we will find a minimum of $\sigma(N,\beta)$ over all $\beta\in\{0,\dots,d-1\}$. To this end, let us denote by $\sigma(n,\beta)$ the number of $\phi\in F_{n,2}$ that fulfil:
\begin{equation}\label{eq_h_beta_4}
\sum_{i=1}^{n}u^{(i)}\phi^{(i)} \mod d=\beta,
\end{equation}
where $n\in \{1,\dots,N\}$. We have introduced a new variable $n$, because in order to find a minimum of $\sigma(N,\beta)$ over $\beta$ we want to express $\sigma(n+1,\beta)$ in terms of $\sigma(n,\beta')$, namely
\begin{equation}\label{eq_sigma_sum}
\sigma (n+1,\beta)=\sigma(n,\beta)+\sigma (n,\beta-u^{(n+1)}),
\end{equation}
however at the end we will set $n=N$. To understand why the above equality is true let us examine (\ref{eq_h_beta_4}), but for $\sigma(n+1,\beta)$:
\begin{eqnarray}\label{eq_h_betax2}
\sum_{i=1}^{n+1}u^{(i)}\phi^{(i)} \mod d&=&\beta,\nonumber\\
\Downarrow \nonumber\\
\sum_{i=1}^{n}u^{(i)}\phi^{(i)} \mod d&=&\beta-u^{(n+1)}\phi^{(n+1)}.
\end{eqnarray}
As we can see, for $\phi^{(n+1)}=0$ we get an identical equation to (\ref{eq_h_beta_4}) and for $\phi^{(n+1)}=1$ we also have an equation equivalent to (\ref{eq_h_beta_4}) but for $\beta=\beta'-u^{(n+1)}$ and so $\sigma(n+1,\beta)$ is a sum of the number of solution for those two cases, therefore we get (\ref{eq_sigma_sum}).

Let us denote by $\sigma_{\min}(n)$ the minimal, nonzero $\sigma(n,\beta)$ over all $\beta \in\{0,\dots,d-1\}$. From (\ref{eq_sigma_sum}) it follows that either
\begin{equation}\label{eq_sigma_n+1_1}
\sigma_{\min}(n+1)= \sigma_{\min}(n)
\end{equation}
or 
\begin{equation}\label{eq_sigma_n+1_2}
\sigma_{\min}(n+1)\geqslant 2\sigma_{\min}(n).
\end{equation}
If (\ref{eq_sigma_n+1_1}) is true, then by the virtue of (\ref{eq_sigma_sum}) it implies that for some $\beta\in\{0,\dots,d-1\}$ for which $\sigma(n,\beta)=\sigma_{\min}(n)$ we have $\sigma (n,\beta-u^{(n+1)})=0$. Let us examine that case further. We can once again use Eq. (\ref{eq_sigma_sum}) to get
\begin{equation}
\sigma(n+1,\beta-u^{(n+1)})=\sigma(n,\beta-2u^{(n+1)}).
\end{equation}
Now, let us assume that $\sigma(n,\beta-k u^{(n+1)})=0$ for all positive integers $k\leqslant \kappa$, where $\kappa\in \mathbb{Z}_{+}$. Then from Eq. (\ref{eq_sigma_sum}) we have
\begin{equation}
\sigma(n+1,\beta-k u^{(n+1)})=\sigma(n,\beta-(k+1)u^{(n+1)}).
\end{equation}
Notice, that our assumption $\sigma(n,\beta-k u^{(n+1)})=0$ gives us a bound on the value of $\kappa$, since 
\begin{equation}
\sigma(n,\beta-d u^{(n+1)})=\sigma(n,\beta)=\sigma_{\min}(n)\neq 0,
\end{equation}
and so $\kappa \in \{1,\dots d-1\}$. Importantly, in the case of $k=\kappa$ it follows from Eq. (\ref{eq_sigma_sum}) that
\begin{equation}\label{eq_sigma_0}
\sigma(n,\beta -\kappa u^{(n+1)})=0
\end{equation}
and
\begin{equation}\label{eq_sigma_neq0}
\sigma(n+1,\beta -\kappa u^{(n+1)})\neq 0.
\end{equation}
In other words, every time (\ref{eq_sigma_n+1_1}) is true, it follows that there exists at least one $\beta'\in\{0,\dots,d-1\}$ for which equation
\begin{equation}\label{eq_tau_beta_cond}
\sum_{i=1}^{n'} u^{(i)}\phi^{(i)} \mod d = \beta'
\end{equation}
has no solution when $n'=n$ (as in (\ref{eq_sigma_0})) and has at least one solution when $n'=n+1$ (as in (\ref{eq_sigma_neq0})).

Finally, we are ready to set a bound on $\sigma_{\min}(n)$. In the case of $n=1$ we have exactly two values of $\beta$ for which we have a solution to Eq. (\ref{eq_h_beta_2}), so there are $d-2$ values of $\beta$ with no solution. If we assume that $\sigma_{\min}(2)=\sigma_{\min}(1)=1$, then we know that there are $d-3$ values of $\beta'$ for which (\ref{eq_h_beta_2}) has no solutions. Following this pattern we conclude that if for all $n\in\{1,\dots,d-2\}$ we have $\sigma_{\min}(n+1)=\sigma_{\min}(n)=1$, then for $n=d-1$ and for all $\beta\in\{0,\dots,d-1\}$ the following holds true
\begin{equation}
\sigma(d-1,\beta)\neq 0.
\end{equation}
From this we can infer that for $n>d-1$, $\sigma_{\min}(n)$ will be bounded according to (\ref{eq_sigma_n+1_2}). Moreover, since in general the bound from (\ref{eq_sigma_n+1_2}) for the case $n'=n+1$ is a linear function of the bound for the case $n'=n$, the exact order of applying bounds (\ref{eq_sigma_n+1_1}) and (\ref{eq_sigma_n+1_2}) does not matter - only the total number of times that each bound is applied matters. Therefore, the assumption $\sigma_{\min}(n+1)=\sigma_{\min}(n)=1$ for $n\in\{1,\dots,d-2\}$ yields a lower bound on on $\sigma_{\min}(n)$:
\begin{equation}\label{eq_sigma_bound}
\sigma_{\min}(n)\geqslant
\begin{cases}
1 \quad \textrm{for} \quad n\leqslant d-1,\\
2^{n-(d-1)} \quad \textrm{for} \quad n> d-1.
\end{cases}
\end{equation}
\and so we can simply write
\begin{equation}
\sigma_{\min}(N)\geqslant 2^{N-(d-1)}
\end{equation}
which ends the proof.
\end{proof}
This leads us strait to Theorem \ref{Raimat}
which for completeness we restate here.

\setcounter{thm}{2}
\begin{thm}\label{thm_ge_ks}
Consider a stabilizer $\mathbb{S}=\langle G_{1},G_{2}\rangle$. If the local dimension $d$ is prime and the corresponding stabilizer subspace $V$ is genuinely entangled then $\dim K(\mathbb{S})\geqslant \lceil \frac{N-1}{d-1} \rceil$.
\end{thm}
\begin{proof}
From the number of generators it directly follows that $\dim K(\mathbb{S})\leqslant 1$, but since we assume that $V$ is genuinely entangled we know that $v_{1,2}\neq 0$, and so we can just write $\dim K(\mathbb{S})= 1$. We need to show that
\begin{equation}\label{eq_1_geq}
1\geqslant \left\lceil \frac{N-1}{d-1} \right\rceil.
\end{equation}
By the virtue of Lemma \ref{lem_ge}, if $V$ is genuinely entangled, then 
\begin{equation}\label{eq_h=0}
h(v_{1,2},\phi)=0
\end{equation}
is only true for the representations of trivial bipartitions $\phi\in\{\phi_{T_{0}},\phi_{T_{1}}\}$, i.e., the number of $\phi\in F_{N,2}$ that fulfil (\ref{eq_h=0}) equals $2$. Then, from Lemma \ref{lem_2^(-(d-1))} we have
\begin{equation}
2\geqslant 2^{N-(d-1)},
\end{equation}
from which it easily follows that
\begin{equation}
1\geqslant \frac{N-1}{d-1}.
\end{equation}
If this inequality is true, then (\ref{eq_1_geq}) is as well, which ends the proof.
\end{proof}
Now that this theorem has been proven, we can discus what are the difficulties of proving Conjecture \ref{con_ge} in its most general form. The issue lays in the proof of a version of Lemma \ref{lem_2^(-(d-1))} for an arbitrary $k$. The number of solutions in the most general case would be $2^{N-(d-1)\cdot \dim K(\mathbb{S})}$, but trying to prove it by using the same arguments runs into problem when discussing the number of $\sigma(n,\beta)$ that are zero for $n=n'$ and are nonzero for $n=n'+1$. For arbitrary $k$, the procedure from Lemma \ref{lem_2^(-(d-1))} would give us the number of $\phi$ to be larger or equal to $2^{N-(d^{\dim K(\mathbb{S})}-1)}$. This is a result of the fact that for higher $\dim K(\mathbb{S})$ a situation where increase in $n$ makes only one $\sigma(n,\beta)$ nonzero is really rare. For example, if increasing $n$ from 0 to $d-1$ resulted in the change from zero to nonzero for $d-1$ sigmas, then each increase after that would make at least $d$ sigmas nonzero (this situation occurs when on the first $d-1$ sites, one vector $v_{i,j}$ has ones and the other vectors have zeros). To conclude, in order to prove Conjecture \ref{con_ge} using the same methods as we did, one would need to find a better way of estimating how many $\sigma(n,\beta)$ become nonzero with each increase in $n$.
\section{Genuinely entangled, stabilizer subspace of maximal dimension}\label{ap_gmax}
In this section we show an example of a genuinely entangled, stabilizer subspace of dimension that we claim is maximal for such subspaces (see Conjecture \ref{con_dim}). More specifically we construct the generators of stabilizer of the aforementioned subspace. Since the expressions on the generators are recursive equations dependent on the number of qubits $N$, we will use a notation $G_{i}(N)$ for an $i$'th generator for $N$ qubits. Similarly, we denote by $v_{i,j}(N)$ the vector defined as in (\ref{eq_vij_def}) for generators $G_{i}(N)$ and $G_{j}(N)$. Moreover, the number of generators $k$ will also depend on $N$, however in this case we will use $k$ and $k(N)$ interchangeably depending on whether the dependence on $N$ is important to a specific case.

Let $k=k_{min}(N)$ be defined by Eq. (\ref{eq_k_min}) and let us consider a set of generators $\{G_{i}(N)\}_{i=1}^{k}$. For $N=2$ we construct the generators as
\begin{eqnarray}\label{eq_gmax_2_d}
G_{1}(2)&&=X\otimes X^{d-1},\nonumber\\
G_{2}(2)&&=Z\otimes Z.
\end{eqnarray}
For arbitrary $N$ we give a recurrence relation for the generators and we distinguish two cases: if $N$ fulfils $k(N-1)=k(N)-1$, then
\begin{equation}\label{eq_gmax_N=_d}
G_{i}(N)=\begin{cases}
G_{i}(N-1)\otimes \mathbb{1} &i\notin \{k-1,k\},\\
G_{i}(N-1)\otimes P_{k-1}^{d-1} & i=k-1,\\
\mathbb{1}^{\otimes(N-2)}\otimes P_{k}\otimes P_{k} &i=k,
\end{cases}
\end{equation}
and otherwise
\begin{equation}\label{eq_gmax_Nneq_d}
G_{i}(N)=\begin{cases}
G_{i}(N-1)\otimes \mathbb{1} &i\notin \{l,k\},\\
\tilde{G}_{i}(N-1)\otimes P_{k-1}^{d-1-m} &i=l,\\
G_{i}(N-1)\otimes P_{k} &i=k,
\end{cases}
\end{equation}
where we define $\tilde{G}_{i}(N-1)$ as
\begin{equation}\label{eq2_gmax_Nneq_d}
\tilde{G}_{i}(N-1)=G_{i}(N-1)\left(\mathbb{1}^{\otimes (N-2)}\otimes P_{k-1}^{m+1}\right),
\end{equation}
parameter $m$ is equal to
\begin{eqnarray}\label{m_prop}
m&=&\sup(\mathcal{M}_{N,d}), \nonumber\\
\mathcal{M}_{N,d}&=&\Bigg\{n\in\mathbb{N}:\left\lceil\frac{N-1-n}{d-1}\right\rceil=\left\lceil\frac{N-1}{d-1}\right\rceil\Bigg\},
\end{eqnarray}
matrix $P_{i}$ equals 
\begin{equation}\label{Pi}
P_{i}=\left\{ 
\begin{array}{ll}
X & \textrm{ for odd\;} i, \\[1ex]
Z & \textrm{ for even\;}i,
\end{array}
\right.
\end{equation}
and
\begin{equation}\label{eq_l_def}
l(N)=\frac{k(k-1)}{2}+1-\left\lceil \frac{N-1}{d-1} \right\rceil.
\end{equation}
The above function $l(N)$ is defined on the set of $N$ for which $k(N)=\operatorname{const}$. Let $\{N_{\min},\dots,N_{\max}\}$ be a set of $N$ for which $k(N)=\operatorname{const}$, i.e.,
\begin{eqnarray}
k(N_{\min}-1)&=k(N_{\min})-1,\nonumber\\
k(N_{\max}+1)&=k(N_{\max})+1.
\end{eqnarray}
From (\ref{eq_n-1/d-1}) we have
\begin{equation}
\left\lceil \frac{N_{\min}-2}{d-1}\right\rceil = \frac{1}{2}[k(N_{\min})-1] [k(N_{\min})-2],
\end{equation}
\begin{equation}
\left\lceil \frac{N_{\max}-1}{d-1}\right\rceil = \frac{1}{2}k(N_{\min})\; [k(N_{\min})-1].
\end{equation}
Substituting them into (\ref{eq_l_def}) we get a set $\{1,\dots,k-1\}$, which is the set of all $l(N)$ for $k=\operatorname{const}.$. Let us note here that for $d=2$ this set is effectively $\{1,\dots,k-2\}$, because in that case equation $l=k-1$ is equivalent to $k(N-1)=k(N)-1$ which is the condition for (\ref{eq_gmax_N=_d}).
\begin{thm}
The stabilizer $\mathbb{S}_{\max}$ generated by the generators defined in (\ref{eq_gmax_2_d}), (\ref{eq_gmax_N=_d}) and (\ref{eq_gmax_Nneq_d}) stabilize a genuinely entangled, stabilizer subspace, where $k_{\min}(N)$ is defined in (\ref{eq_k_min}). Moreover, if $d$ is prime then the dimension of the stabilizer subspace equals $d^{N-k_{\min}(N)}$.
\end{thm}
Before we start the proof of the above theorem, let us introduce a convenient convention, namely for even $k$ instead of calculating vectors $v_{i,k}$ we will calculate
\begin{equation}
v'_{i,k}=(d-1)v_{i,k}.
\end{equation}
This makes it so that the vast majority of vector elements equals $0,1$, which simplifies our calculations. Moreover, this does not pose any problems, since in this proof we only care about the general properties of $K(\mathbb{S}_{\max})$, which both $v_{i,k}$ and $v'_{i,k}$ are a part of, and not the vectors in particular. In practise, it is as if we assumed that $k$ is always odd while calculating the elements of $v_{i,k}$. With that, we can move to the main part of the proof.
\begin{proof}
To complete the proof we first need to show that $\mathbb{S}_{\max}=\langle G_{1}(N),\dots,G_{k}(N)\rangle$ is the stabilizer, after which we need to prove that the stabilizer subspace is genuinely entangled. The subgroup $\mathbb{S}$ is a stabilizer if $a\mathbb{1}\notin\mathbb{S}$ for $a\neq 1$ and if all the elements of $\mathbb{S}$ commute. This can be shown by proving the following conditions
\begin{equation}\label{eq_stab_nont_1}
\forall_{i\in\{1,\dots,k\}}\quad G_{i}(N)^{d}=\mathbb{1},
\end{equation}
\begin{equation}\label{eq_stab_nont_2}
\forall_{i,j\in\{1,\dots,k\}}\quad \left[G_{i}(N),G_{j}(N)\right]=0,    
\end{equation}
\begin{equation}\label{eq_stab_nont_3}
\forall_{a\neq 1}\;\forall{\alpha_{i}\in\{0,\dots,d-1\}}\quad \prod_{i=1}^{k}G_{i}^{\alpha_{i}}(N)\neq a \mathbb{1}.
\end{equation}
A proof of (\ref{eq_stab_nont_1}) is trivial, since by definition the generators $G_{i}(N)$ are a tensor product of matrices \eqref{eq_xz_def}. 

Next, proving (\ref{eq_stab_nont_3}) is also fairly simple: from Eq. (\ref{eq_gmax_2_d}) and Eq. (\ref{eq_gmax_N=_d}) it follows that $G^{(1)}_{i}(N)=X$ implies $i=1$ and similarly  $G^{(2)}_{i}(N)=Z$ implies $i=2$. Moreover, from Eq. (\ref{eq_gmax_N=_d}) and Eq. (\ref{eq_gmax_Nneq_d}) we can conclude that for $n$ for which $k(n)=k(n+1)-1$, $G_{i}^{(n)}(N)=P_{k(n)}$ implies $i=k(n)$. Consequently, for every generator $G_{i}(N)$ we can find $n$ such that $G_{i}^{(n)}(N)\neq G_{j}^{(n)}(N)$ for $i\neq j$, which in turn implies that
\begin{equation}
\prod_{i=1}^{k}G_{i}^{\alpha_{i}}(N)\sim \mathbb{1},
\end{equation}
iff for all $i\in \{1,\dots, k\}$, $\alpha_{i}=0$, but for that case $\prod_{i=1}^{k}G_{i}^{0}(N)=\mathbb{1}$, which proves (\ref{eq_stab_nont_3}).

Proving Eq. (\ref{eq_stab_nont_2}) is more complex, but thankfully we can make use of the vectors $v_{i,j}(N)$ defined in Eq. (\ref{eq_vij_def}). In this formalism, the condition (\ref{eq_stab_nont_2}) translates to
\begin{equation}\label{eq_vij_sum_d}
\sum_{n=1}^{N}v_{i,j}^{(n)}=0 \mod d.
\end{equation}
Let us begin with $N=2$. From Eq. (\ref{eq_gmax_2_d}) it is clear that $v_{1,2}(2)=(1,d-1)$ and so (\ref{eq_vij_sum_d}) is fulfilled. The proof for an arbitrary $N$ will be done by induction, i.e., we assume that (\ref{eq_vij_sum_d}) is fulfilled for $N'=N-1$ and we show that it is fulfilled for $N'=N$. For clarity, we will consider two separate cases: first (\ref{eq_gmax_N=_d}) and then (\ref{eq_gmax_Nneq_d}).

From (\ref{eq_gmax_N=_d}) we can easily conclude that for $i,j\neq k$ we have
\begin{equation}\label{eq_vij_N=_d}
v_{i,j}(N)=v_{i,j}(N-1)\oplus (0).
\end{equation}
As for the vectors $v_{i,k}(N)$, from (\ref{eq_gmax_N=_d}) it is clear that 
\begin{eqnarray}
v_{i,k}^{(N)}(N)=0& \qquad &i\neq k-1,\nonumber\\
v_{i,k}^{(N)}(N)=d&-1\qquad &i=k-1,\nonumber\\
v_{i,k}^{(n)}(N)=0& &n\in\{1,\dots,N-2\}.
\end{eqnarray}
The only unknown elements are $v_{i,k}^{(N-1)}$, however they also can be easily determined, by analysing generators for $N'=N-1$. There are two subcases: for $N=3$ we have to refer to (\ref{eq_gmax_2_d}) and for $N>3$ to (\ref{eq_gmax_Nneq_d}). However, according to both (\ref{eq_gmax_2_d}) and (\ref{eq_gmax_Nneq_d}) we have
\begin{eqnarray}
G_{i}^{(N-1)}(N-1)&\in& \{\mathbb{1},P_{k'-1}^{q}\}\quad  i\in\{1,\dots,k'-1\},\nonumber\\
G_{k'}^{(N-1)}(N-1)&=&P_{k'},
\end{eqnarray}
where $k'=k(N-1)=k-1$ and $q\in \mathbb{Z}_{+}$. This implies that
\begin{eqnarray}\label{eq_vik_d}
v_{i,k}(N)&=&0  \qquad i\in\{1,\dots,k-2\},\nonumber\\
v_{k-1,k}(N)&=&e_{N-1}+(d-1)e_{N}.
\end{eqnarray}
Clearly from induction, all vectors $v_{i,j}(N)$ fulfil (\ref{eq_vij_sum_d}).

Next, let us consider the case (\ref{eq_gmax_Nneq_d}). It is easy to see that
\begin{equation}\label{eq_vij_+0_d}
v_{i,j}(N)=v_{i,j}(N-1)\oplus (0)
\end{equation}
is true for $i,j\neq l$. Moreover, the above is also true for $i\neq l,k$ and $j=l$, since from (\ref{eq_gmax_N=_d}) and (\ref{eq_gmax_Nneq_d}) it follows that for such $i$ we have
\begin{equation}
G_{i}^{(N-1)}(N-1)\in \left\{\mathbb{1},P_{k-1}^{q}\right\}
\end{equation}
for some $q\in \mathbb{Z}_{+}$. This means that we have one vector left to consider, namely for $i=l$ and $j=k$ we have
\begin{equation}\label{eq_vij_eN_d}
v_{l,k}(N)=v_{l,k}(N-1)\oplus (0) +(m+1)e_{N-1} +(d-1-m)e_{N}.
\end{equation}
Clearly, if for $N'=N-1$ (\ref{eq_vij_sum_d}) is fulfilled then all $v_{i,j}(N)$ also fulfil (\ref{eq_vij_sum_d}). This proves (\ref{eq_stab_nont_2}) for all $N$.

To reiterate, we proved that \eqref{eq_stab_nont_1}, \eqref{eq_stab_nont_2} and \eqref{eq_stab_nont_3} are fulfilled by $G_{1},\dots,G_{k}$, hence $\mathbb{S}_{\max}$ is a stabilizer. Moreover, since we have shown that $\prod_{i=1}^{k}G_{i}^{\alpha}(N)=\mathbb{1}$ is only true if all $\alpha_{i}=0$, the generators are independent and so if $d$ is prime then the dimension of the stabilizer subspace equals $d^{N-k_{\min}(N)}$ \cite{GHEORGHIU2014505}. This leaves only the question of genuine entanglement of the stablizer subspace $V$.

For $N=2$ (\ref{eq_gmax_2_d}) it is easy to see from Corollary \ref{cor_ge} that $V$ is genuinely entangled, since the only nontrivial bipartition is $\{1\}|\{2\}$. For arbitrary $N$ we again use a prove by induction, i.e., we assume that for $N'=N-1$, the subspace $V(N-1)$ is genuinely entangled and we show that this implies that $V(N)$ is also genuinely entangled. 

First, consider a case $k(N-1)=k(N)-1$, which corresponds to generators (\ref{eq_gmax_N=_d}). From our assumption and from (\ref{eq_vij_+0_d}) it follows that if we only consider nontrivial bipartitions of $\{1,\dots,N-1\}$, then for all such bipartitions there exist $i,j\in\{1,\dots,k-1\}$ such that
\begin{equation}\label{eq_h_neq0}
h(v_{i,j}(N),\phi)\neq 0,
\end{equation}
where $\phi$ is a representation of a bipartition $Q|\overline{Q}$ and $h(\cdot,\cdot)$ is defined in Eq. (\ref{eq_h_def}). Moreover, from (\ref{eq_vij_N=_d}) and (\ref{eq_vij_+0_d}) it follows that $v_{i,j}^{(N)}(N)=0$ for $i,j\neq k$, and so by the above argument, for all $Q\subset \{1,\dots,N\}$ such that $Q$ is nontrivial and $Q\neq \{N\},\{1,\dots,N-1\}$ there exist $i,j\in\{1,\dots,k-1\}$ for which (\ref{eq_h_neq0}) is fulfilled. As for $Q= \{N\},\{1,\dots,N-1\}$, from (\ref{eq_vik_d}) we have that
\begin{equation}
h(v_{k-1,k}(N),\phi)\neq 0,
\end{equation}
hence for all nontrivial bipartitions of the set $\{1,\dots,N\}$ there exist a pair $i,j\in\{1,\dots,k\}$ such that 
\begin{equation}
h(v_{i,j}(N),\phi)\neq 0,
\end{equation}
which by the virtue of Lemma \ref{lem_ge} shows that $V(N)$ is genuinely entangled.

Next, let us consider a case (\ref{eq_gmax_Nneq_d}) for $l(N)=l(N-1)+1$. Since $l(N)$ is a nonincreasing function, we know that $v_{l,k}(N-1)=0$. Then, we can use the same argumentation as for the previous case but with a vector $v_{l,k}(N)$ instead of $v_{k-1,k}(N)$.

Lastly, let us consider a case (\ref{eq_gmax_Nneq_d}) for $l(N)=l(N-1)$. By iterating (\ref{eq_vij_eN_d}) from $m'=m$ to $m'=0$ we can derive an explicit formula for $v_{l,k}(N)$:
\begin{equation}\label{eq_vlk_sum}
v_{l,k}(N)=\sum_{m'=0}^{m} e_{N-1-m'}+(d-1-m)e_{N}.
\end{equation}
Moreover, the same iteration of (\ref{eq_vij_+0_d}) implies that for all pairs $(i,j)\neq (l,k)$ and for all $n\in \{N-m,\dots,N\}$
\begin{equation}
v_{i,j}^{(n)}=0.
\end{equation}
Therefore, our induction assumption implies that for every nontrivial subset $Q\subset \{1,\dots N-1-m\}$ there exist a pair $(i,j)\neq (l,k)$ for which (\ref{eq_h_neq0}) is fulfilled, which by the virtue of Lemma \ref{lem_ge} implies, that $V(N)$ is genuinely entangled. This leaves us with bipartitions $Q|\overline{Q}$ for which $Q\subset \{N-m,\dots,N\}$ or $\overline{Q}\subset \{N-m,\dots,N\}$. However, from (\ref{eq_vlk_sum}) it follows that for every such $Q$ (or $\overline{Q}$), (\ref{eq_h_neq0}) is fulfilled by $v_{l,k}(N)$, which ends the proof.
\end{proof}


\end{document}